\documentclass[11pt]{amsart}
\usepackage{amsmath}

\newcommand{\PP}{\mathbb{P}}
\newcommand{\remove}[1]{}
\newtheorem{lem}{Lemma}[section]
\newtheorem{cor}[lem]{Corollary}
\newtheorem{prop}[lem]{Proposition}
\newtheorem{thm}[lem]{Theorem}

\begin{document}
\title{Tree split probabilities determine the branch lengths}

\author{Benny Chor \and Mike Steel}
\keywords{phylogenetic tree reconstruction, evolutionary model, Markov process, Hadamard transform, systems of polynomial equations, inverse function theorem.\newline 
Benny Chor: School of Computer Science, Tel-Aviv University, Tel-Aviv, Israel. {\tt benny@cs.tau.ac.il} \\
Mike Steel: Mathematics and Statistics Dept., University of Canterbury, Christchurch, New Zealand. {\tt mike.steel@canterbury.ac.nz}}

\maketitle
\begin{abstract}
The evolution of aligned DNA sequence sites is generally modeled by a Markov process operating along
the edges of a phylogenetic tree.  It is well known that the probability distribution on the site patterns at the tips of the tree determines the tree and its branch lengths.
However, the number of patterns is typically much larger than the number of edges, suggesting considerable redundancy in the branch length estimation.  In this paper we ask whether the probabilities of just the `edge-specific' patterns (the ones that correspond to a change of state on a single edge) suffice to recover the branch lengths of the tree, under
a symmetric 2-state Markov process. 
  We first show that this holds provided
the branch lengths are sufficiently short, by applying the inverse function theorem. We then consider whether this restriction to short branch lengths
is necessary, and show that for trees with up to four leaves it can be lifted. This leaves open the interesting question of whether this holds in general.
\end{abstract}

\section{Background}

When a discrete character evolves on a tree under a Markov process, the probability distribution on site
patterns at the leaves of the tree is determined by  the tree and its branch lengths \cite{fels, sem}.  What is less obvious is that this process is invertible for many models -- that is, 
the probability distribution on site
patterns at the leaves uniquely identifies both the tree and its branch lengths.  

This fundamental property underlies
all statistical approaches for inferring evolutionary relationships from aligned genetic sequence data.   In this setting, the `discrete character' refers to the pattern of nucleotides across the species
at each site, and the frequency of this pattern across the sequences provides some estimate of the probability of that pattern. In this paper we are interested in
what the probability distribution says about the branch lengths of the underlying tree (we will assume this topology is known).
Notice that 
the number of site patterns
grows exponentially with the number $n$ of leaves, yet the number of branches of the tree (for which the branch lengths are being
estimated) grows linearly with $n$.   

This suggests a basic question --- do we need all the site pattern probabilities to infer the branch lengths? More precisely,  if a tree has $k$ edges (branches), are there $k$ site patterns whose probabilities under the model might identify the lengths of these branches?  

One motivation for this question is that in practice many site patterns will simply never occur (indeed most will not, if our sequence length grows
at most polynomially with $n$, since the number of site patterns grow exponentially with $n$).  And this is a problem if we try to estimate pattern probabilities from their relative frequency.

There is an natural candidate for  a particular choice of $k$ site pattens --  for each edge we take the site pattern in which all the  leaves on one side of the edge are in one state, and all the
leaves on the other side of the edge are in a different state -- we refer to such a site pattern as a {\em tree split} for this edge.   From a practical perspective, the tree splits are patterns that are likely 
to be observed in the data, since they require just one change of state in the tree.  They also correspond to the primary divisions of the species into two groups (e.g. vertebrates vs invertebrates) and so have a clear phylogenetic meaning.

The question of whether the tree split probabilities determine the branch lengths is a delicate one - we prove that for the 2-state symmetric model, 
answer is yes for 4-leaf binary trees and for $n$-leaf star trees, and we conjecture that it holds true for arbitrary phylogenetic trees.  Our approach exploits the Hadamard representation for the 2-state model \cite{hen, HP93}, as well as computational (symbolic)
algebraic analysis tools.

\section{Model and Notations}
In the Neyman 2-state model~\cite{neyman1971}, each
character admits one out of two states, e.g., purines and pyrimidines. 
Without loss of generality, we denote these
states by 0 and 1. We use the symmetric Poisson model,
where for each edge $e$ of the tree T, there is a corresponding
probability $p_{e}$ $(0\leq p_{e} < 1/2)$ that the character states at
the two incident vertices of $e$ differ, and this probability
is independent of the state at the initial vertex. The probability
$p_{e}$ is the probability of having an odd (1, 3, 5)
number of substitutions per site across the edge $e$. The expected
number of substitutions per site across the edge $e$ equals $q_{e} =
-\frac{1}{2}\ln(1 - 2p_{e})$. The value $q(e)$ is referred to the {\em (branch) length} (or {\em weight}) of edge $e$.
Measuring the tree edges by $q_{e}$ $(0\leq q_{e}<\infty)$, we get an additive measure on the tree, namely,  the expected
number of substitutions between each pair of leaves
(because expected values are additive).   Such a weighted phylogenetic tree is a
probabilistic model that emits any given pattern of states at its
leaves with a well defined probability.   Notice that the limits $q(e) \rightarrow 0$ and $q(e) \rightarrow  \infty$ correspond to the limits
$p(e) \rightarrow 0$ and $p(e) \rightarrow \frac{1}{2}$, respectively.

The observed
sequences at the leaves can be represented by a matrix, $\psi$, where the
number of rows equals the number of species ($m=4$ in our
case), and the number of columns equals the common
length of the sequences ($n$ in our case). For 2-state characters, it is convenient
to ``summarize" the observed data $\psi$ by a vector of observed frequencies of splits, ${\hat{\bf s}}$. This vector simply counts
how many sites share any specific pattern. Under a fully
symmetric 2-state model, the probability of a pattern is equal
to that of its complement (where all 0 and 1 are interchanged). We make the
following convention about indexing the patterns obtained in the sequences over $m=4$ species,
labelled 1, 2, 3, and 4, with the sequences $x_1,x_2,x_3,x_4\in\{0,1\}^n$: We identify a site pattern by the
subset of species {1, 2, 3} whose character at that site is
different from that of species 4.  More generally (i.e. for any value of $m$)  for every
$\alpha\subseteq\{1,\ldots,m-1\}$, an {\em $\alpha$-split pattern} is a pattern where
all taxa in the subset $\alpha$ have one character (0 or 1), and the taxa in the
complement subset have the second character (there are two such patterns). The value
${\hat{s}_\alpha}$ equals the number of times that $\alpha$-split patterns
appear in the data. For $m = 4$ there are $2^{3}=8$ possible
patterns, and the vector of observed sequence frequencies
is ${\hat{\bf
s}}=[{\hat{s}}_{\emptyset},{\hat{s}}_1,{\hat{s}}_2,{\hat{s}}_{12},{\hat{s}}_3,{\hat{s}}_{13},{\hat{s}}_{23},{\hat{s}}_{123}]$.

\section{The tree split probabilities determine the branch lengths locally}

In this section, we show that the multivariate inverse function theorem implies that branch lengths can be recovered from tree split probabilities provided the branch lengths are not
too large. 

\begin{thm} 
\label{mainthm}
Let $T$ be any phylogenetic tree, on any number of leaves, such that the degree of any internal vertex of $T$ is greater than $2$.
In the 2-state symmetric model, the probabilities of the tree splits determine the branch lengths of $T$ in some neighborhood of 
the origin. That is, there exist some $\eta>0$, such that if the branch lengths all lie within $[{0}, {\bf \eta})$,
then they can be uniquely recovered from the tree split probabilities.  
\end{thm}

\begin{proof} 

To simplify notation in this section, given a phylogenetic tree $T$ with $N$ edges, label the edges $e_1, e_2, \ldots, e_N$.
For each $i \in \{1, \ldots, N\}$, let  $\alpha_i$ denote the tree split corresponding to $e_i$;  let $s_i$  be the probability of generating the pattern $\alpha_i$ on $T$ under the symmetric 2-state model; let 
$q_i$ be the branch length of edge $e_i$,  and let $p_i = \frac{1}{2}(1-\exp(-2q_i))$, which is the probability of a change of state on edge $e_i=(v_1^i, v_2^i)$. Consider the two subtrees of  $T$
which result from removing the edge $e_i$ (but not the nodes $v_1^i, v_2^i$). Let $T_1, T_2$ denote the resulting subtrees, rooted at $v_1^i, v_2^i$, respectively. Let $Q_i^1$  be the probability of the event
``all leaves of $T_1$ are in the same state as $v_1^i$'', and $Q_i^2$  be the probability of the event
``all leaves of $T_2$ are in the same state as $v_2^i$''. Let $R_i^1$ denote the probability of the event ``all leaves of $T_1$ are 
in the same state and they differ from the state of $v_1^i$'', and $R_i^2$ 
denote the probability of the event ``all leaves of $T_2$ are in the same state and they differ from the state of $v_2^i$''.
We note that under the 2-state symmetric model, changes of state on different edges are independent events.
By considering whether or not there is a change of state on edge $e_i$, the following identity holds for all $i$: 

\begin{equation}
\label{sq}
s_i = p_iQ^1_i Q^2_i+ (1-p_i)(Q^1_iR^2_i+R^1_iQ^2_i),
\end{equation}

Note that $Q^1_i, Q^2_i,R^1_i, R^2_i$ involve only $p_j$ values for $j \neq i$, and that when all the $p_j$ values are zero we have:

\begin{equation}
\label{eqR}
R^1_i|_{{\bf p=0}} =R^2_i|_{{\bf p=0}}=0, \mbox{ and } Q^1_i|_{{\bf p=0}} =Q^2_i|_{{\bf p = 0}}=1.
\end{equation}

Now, consider the Jacobian  matrix $${\bf J} = \left[ \frac{\partial s_i}{\partial p_j}\right].$$
From Eqn. (\ref{sq}) and the fact that $p_i$ does not appear in $Q_i^1, Q_i^2$ and $R_i^1, R_i^2$ we have:
$$\frac{\partial s_i}{\partial p_i} =  Q^1_i Q^2_i- (Q^1_iR^2_i+R^1_iQ^2_i)$$
and from Eqn. (\ref{eqR}) this equals $1$ when ${\bf p=0}$.

Similarly, for $j \neq i$, Eqn. (\ref{sq}) gives:

$$\frac{\partial s_i}{\partial p_j} = p_i \frac{\partial}{\partial p_j} Q^1_i Q^2_i + (1-p_i) \frac{\partial}{\partial p_j}(Q^1_iR^2_i+R^1_iQ^2_i)\ ,$$
and so, when ${\bf p = 0}$ the first term on the right vanishes (since $p_i = 0)$ and we have:
$$\frac{\partial s_i}{\partial p_j}|_{{\bf p=0}} =  \frac{\partial}{\partial p_j}(Q^1_iR^2_i+R^1_iQ^2_i)|_{{\bf p=0}}.$$

Now, notice that $Q^1_iR^2_i+R^1_iQ^2_i$ is a multinomial polynomial of the variables $p_k$, $k \neq i$.   We argue
that this polynomial has no term of the form $cp_j$ for a constant $c \neq 0$.  Suppose otherwise, then setting $p_j=\frac{1}{4}$ and $p_k =0$ for all
$k \neq j$, Eqn. (\ref{sq}) shows that the pattern $\alpha_i$ occurs with probability $\frac{1}{4}c \neq 0$ under such an edge probability assignment.
However, since there is no vertex of degree 2 in $T$, edge $e_i$ is the only edge for which $\alpha_i$ has positive probability when
all but one $p-$value is set to zero. This  contradicts the assumption that we are in the setting where $j \neq i$.
Since $Q^1_iR^2_i+R^1_iQ^2_i$ has no term of the form $cp_j$ ($c\neq 0$)  it follows that:   
$$ \frac{\partial}{\partial p_j}(Q^1_iR^2_i+R^1_iQ^2_i)|_{{\bf p=0}}=0.$$

Summarizing, we have $\frac{\partial s_i}{\partial p_i}|_{{\bf p=0}} = \delta_{ij}$ (Kronecker delta), and so at 
${\bf p=0}$, ${\bf J}$ is the $N \times N$ identity matrix. 
Since the map ${\bf p} \mapsto (s_1, \ldots, s_N)$ is a polynomial map, and therefore continuously differentiable, 
and since the Jacobian of this map is invertible at ${\bf p}= {\bf 0}$, the conditions for the multivariate Inverse 
Function Theorem apply at this point. 
Thus, for some neighborhood neighborhood of ${\bf p=0}$, 
the function $(p_1, \ldots, p_N) \mapsto (s_1, \ldots, s_N)$ is invertible. 
Finally, observe that the map from $[0, \infty)^N$ onto $[0, 1/2)^N$ defined by $(q_1, \ldots q_N) \mapsto (p_1, \ldots p_N)$ 
is invertible, and so the theorem  now follows. 
\end{proof}

Theorem~\ref{mainthm} begs the question: how large can $\eta$ be in order for invertibility to hold?  
 In the next section we show that we can obtain invertibility with $\eta=\infty$ for trees with at most four leaves.

\section{Exact analysis for small trees}
For phylogenetic trees with two and three leaves it is easily verified that the tree splits determine the branch lengths with no restriction required on the size of these branch lengths.
Thus we will consider trees with four leaves, of which there are two types - the resolved binary tree, and the star tree. 
We will show that the probability distribution on tree splits
determines the branch lengths of the binary tree for all strictly positive branch lengths, and then establish that the same holds for the star tree. 

 A useful tool in this analysis is the Hadamard representation of the 2-state symmetric model, which we first recall. 

\subsection{Hadamard representation}

 Given a tree T with $m$ leaves and edge
probabilities ${\bf p}=[p_e]_{e\in E(T)}(0\leq p_e< \frac{1}{2})$, the probability of
generating an $\alpha$-split pattern ($\alpha\subseteq\{1,\ldots,m-1\}$) is
determined (and is equal for all sites). Denote this probability by
${\bf s_{\alpha}}=Pr(\alpha-split\mid T,p)$.

Using the same indexing scheme as
above,when $m=4$, we define the vector of pattern generation probabilities ${\bf
s}=[s_\emptyset,s_1,s_2,s_{12},s_3,s_{13},s_{23},s_{123}]$. This vector is
termed the expected sequence spectrum in~\cite{HP93}, where the indexing scheme is explained as well.

\begin{figure}[h]
\begin{center}\unitlength=1mm
\begin{picture}(30,30)
\put(18,15){\line(1,0){14}}\put(23,12.5){$q_{12}$}
\put(-12,5){\line(3,1){30}}\put(0,11.5){$q_{4}$}
\put(-12,25){\line(3,-1){30}}\put(0,18){$q_3$}
\put(32,15){\line(1,1){5}}\put(35.5,17){$q_1$}
\put(32,15){\line(1,-1){5}}\put(35.5,12.5){$q_2$}
\put(-16,25){$x_3$}
\put(-16,3){$x_4$}
\put(39,21){$x_1$}
\put(39,8){$x_2$}
\end{picture}\\
\end{center}
\caption{Example of a resolved tree with four leaves}
\label{treeEx1}
\end{figure}
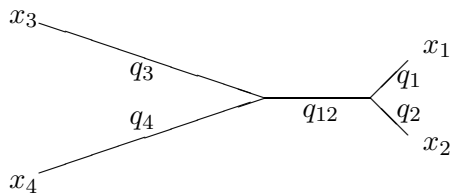
Even though in principle, the edge lengths $q_e=[q_e]_{e\in E(T)}$ determine the
vector $\bf s$ of pattern generation probabilities, it is not obvious how to
actually compute this vector, given the edge lengths. This
is where the Hadamard conjugation \cite{HP93, HPS94}
 comes in. This
transformation yields a powerful tool, which greatly simplifies and
unifies the analysis of phylogenetic data.\\

\noindent
{\bf Definition:}
A {\em Hadamard matrix} of order $\ell$ is an $\ell\times \ell$ matrix $A$ with
$\pm 1$ entries such that $A^tA=\ell I_\ell$, where $I_\ell$ is the identity $\ell\times\ell$ matrix.\\

We will use a special family of Hadamard matrices, 
whose sizes, $\ell$, are powers of 2 \cite{mac77}, defined inductively for
$m\geq 0$ by $H_0=[1]$ and 
$H_{m+1}=\left[ \begin{array}{cc}
                         H_m & H_m\\
                         H_m & -H_m
                 \end{array}\right]$.
It is convenient to index the rows and columns of $H_m$ by
lexicographically ordered subsets of
$\{1,\ldots,m\}$.  Denote by $h_{ \alpha \gamma }$ the $(\alpha,\gamma)$ entry of
$H_m$, then $h_{\alpha \gamma}=(-1)^{|\alpha \cap \gamma |}$. This implies
that $H_m$ is symmetric, namely $H_m^t=H_m$, and thus by the
definition of Hadamard matrices $H_m^{-1}=2^{-m}H_m$.\\
\medskip

\begin{prop} \cite{HP93}
\label{Hadam}
 If ${\bf p}<{\bf 1/2}$ then
${\bf s} = {\bf s(q)} = H^{-1}\exp(H{\bf q})$ where $H=H_{m-1}$, 
namely for $\alpha\subseteq\{1,\ldots$ $,m-1\}$, 
$${\bf s}_\alpha=2^{-(m-1)}\sum_{\gamma}h_{\alpha \gamma}\left(\exp\left(\sum_{\delta }  
h_{\gamma \delta}q_\delta\right)\right).$$
\end{prop}

\medskip

We note that the transformation is reversible, so if $H{\bf s}>{\bf 0}$ (all entries are positive) then
${\bf q} = {\bf q(s)} = H^{-1}\ln(H{\bf s})$. Thus, when $m=4$, the eight  components in the expected sequence spectrum
uniquely determine the five edge weights of the corresponding four taxa tree. 

The question we explore is if the five components that correspond to splits in the tree, namely $s_1, s_2, s_{12}, s_3, s_{123}$ also determine the  five edge weights $q_1, q_2, q_{12}, q_3, q_{123}$. In other words, is the mapping 
$(q_1, q_2, q_{12}, q_3, q_{123}) \mapsto (s_1, s_2, s_{12}, s_3, s_{123})$ one to one?

\subsection{Resolved Tree on Four Leaves}
In this case, the tree has five edges with positive lengths, $0\leq q_{\alpha}<\infty$.
To describe the mapping $(q_1, q_2, q_{12}, q_3, q_{123}) \mapsto (s_1, s_2, s_{12}, s_3, s_{123})$,  
let us first denote $$a_{1}=e^{-2q_{1}}, a_{2}=e^{-2q_{2}}, a_{12}=e^{-2q_{12}}, a_{3}=e^{-2q_{3}}, a_{123}=e^{-2q_{123}},$$ where
the $a_{\alpha}$ are in the interval $(0,1]$ (corresponding to 
edge lengths $0\leq q_{\alpha}<\infty$). Then by the Hadamard transform (with the 8-by-8 matrix, $H_3$),
\begin{align*}
s_1&= a_{123}a_{3} - a_{1}a_{2} + a_{12}a_{2}a_{3} + a_{12}a_{123}a_{2} - a_{1}a_{12}a_{3} - a_{1}a_{12}a_{123} - a_{1}a_{123}a_{2}a_{3} + 1\\
s_2&= a_{123}a_{3} - a_{1}a_{2} - a_{12}a_{2}a_{3} - a_{12}a_{123}a_{2} + a_{1}a_{12}a_{3} + a_{1}a_{12}a_{123} - a_{1}a_{123}a_{2}a_{3} + 1\\
s_{12}&= a_{123}a_{3} + a_{1}a_{2} - a_{12}a_{2}a_{3} - a_{12}a_{123}a_{2} - a_{1}a_{12}a_{3} - a_{1}a_{12}a_{123} + a_{1}a_{123}a_{2}a_{3} + 1\\
s_3&= -a_{123}a_{3} + a_{1}a_{2} - a_{12}a_{2}a_{3} + a_{12}a_{123}a_{2} - a_{1}a_{12}a_{3} + a_{1}a_{12}a_{123} - a_{1}a_{123}a_{2}a_{3} + 1\\
s_{123}&= -a_{123}a_{3} + a_{1}a_{2} + a_{12}a_{2}a_{3} - a_{12}a_{123}a_{2} + a_{1}a_{12}a_{3} - a_{1}a_{12}a_{123} - a_{1}a_{123}a_{2}a_{3} + 1
\end{align*}

\noindent
To show that  the mapping $(q_1, q_2, q_{12}, q_3, q_{123}) \mapsto (s_1, s_2, s_{12}, s_3, s_{123})$ is indeed one to one, it suffices to 
show the following: For any assignment in the interval $(0,1]$ of the ten  variables satisfying $(a_1,a_2,a_{12},a_3,a_{123}), (a'_1,a'_2,a'_{12},a'_3,a'_{123})$, if
\begin{align*}
&s_{1}(a_1,a_2,a_{12},a_3,a_{123})=s_{1}(a'_1,a'_2,a'_{12},a'_3,a'_{123}), \\
&s_{2}(a_1,a_2,a_{12},a_3,a_{123})=s_{2}(a'_1,a'_2,a'_{12},a'_3,a'_{123}), \\
&s_{12}(a_1,a_2,a_{12},a_3,a_{123})=s_{12}(a'_1,a'_2,a'_{12},a'_3,a'_{123}), \\
&s_{3}(a_1,a_2,a_{12},a_3,a_{123})=s_{3}(a'_1,a'_2,a'_{12},a'_3,a'_{123}), \\ 
&s_{123}(a_1,a_2,a_{12},a_3,a_{123})=s_{123}(a'_1,a'_2,a'_{12},a'_3,a'_{123}),
\end{align*}
then we have $(a_1,a_2,a_{12},a_3,a_{123})=(a'_1,a'_2,a'_{12},a'_3,a'_{123})$. As the mapping from $(q_1, q_2, q_{12}, q_3, q_{123})$
 to $(a_1, a_2, a_{12}, a_3, a_{123})$  is clearly one to one, this will establish the desired uniqueness.
 
To simplify the polynomial system, we first preprocess it by computing a few linear combinations of the $s_\alpha$
(we remark that the preprocessing leads to a substantial saving subsequently).
\begin{align*}
&(s_{1}+s_{2})/2=-a_1a_{123}a_2a_3 - a_1a_2 + a_{123}a_3 + 1\\
&(s_{1}-s_{3})/2=
-a_1a_{12}a_{123} + a_{12}a_2a_3 - a_1a_2 + a_{123}a_3\\
&(s_{2}+s_{12})/2 =-a_{12}a_{123}a_2 - a_{12}a_2a_3 + a_{123}a_3 + 1\\
&(s_{1}+s_{2}+s_{3}+s_{123})/4=-a_1a_{123}a_2a_3 + 1
\end{align*}

These equalities give rise to the following set of five polynomial equations in ten variables. 
\begin{align*}
\mbox{\tt (eq1)\ \ \ } & a_{123}a_{3} - a_{1}a_{2} + a_{12}a_{2}a_{3} + a_{12}a_{123}a_{2} - a_{1}a_{12}a_{3} - a_{1}a_{12}a_{123} - a_{1}a_{123}a_{2}a_{3} \\
&= a'_{123}a'_{3} - a'_{1}a'_{2} + a'_{12}a'_{2}a'_{3} + a'_{12}a'_{123}a'_{2} - a'_{1}a'_{12}a'_{3} - a'_{1}a'_{12}a'_{123} - a'_{1}a'_{123}a'_{2}a'_{3} \\
\mbox{\tt (eq2)\ \ \ }  &-a_1a_{123}a_2a_3 - a_1a_2 + a_{123}a_3=-a'_1a'_{123}a'_2a'_3 - a'_1a'_2 + a'_{123}a'_3 \\
\mbox{\tt (eq3)\ \ \ }  &-a_1a_{12}a_{123} + a_{12}a_2a_3 - a_1a_2 + a_{123}a_3=-a'_1a'_{12}a'_{123} + a'_{12}a'_2a'_3 - a'_1a'_2 + a'_{123}a'_3 \\
\mbox{\tt (eq4)\ \ \ }  &-a_{12}a_{123}a_2 - a_{12}a_2a_3 + a_{123}a_3=-a'_{12}a'_{123}a'_2 - a'_{12}a'_2a'_3 + a'_{123}a'_3 \\
\mbox{\tt (eq5)\ \ \ }  &-a_1a_{123}a_2a_3=-a'_1a'_{123}a'_2a'_3 \\
\end{align*}
Even after the simplification, this system is too involved to solve manually. However, it is amenable to being solved using symbolic computer algebra
packages. Specifically, we employed the mathematical symbolic package Maple. We express the system as a system in 5 variables, $a'_1,a'_2,a'_{12},a'_3,a'_{123}$, so Maple tries to solve for these, viewing $a_1,a_2,a_{12},a_3,a_{123}$ as given.\\
\smallskip
\noindent
{\tt > with(SolveTools):}\\
{\tt > Sols := PolynomialSystem([eq{1}, eq{2}, eq{3}, eq4, eq5],}\\ 
\mbox{\ \ \ \ \ \ \ \ \ \ \ \ \ \ \ \ \ \ \ }{\tt \{a1',a2',a12',a3',a123'\}, domain = real):}\\
\noindent
Maple (version 17) produced 5 sets of solutions  to this system. 
The first solution,
$$\{a'_1 = a_1, a'_{12} = a_{12}, a'_{123} = a_{123}, a'_2 = a_2, a'_3 = a_3\}\ ,$$ is exactly what we are looking for. Still, we have to show that the other four solutions are illegal, namely at least one of the variables is out of the interval $(0,1]$.

The 2nd, 3rd, and 4th solutions are 
\begin{align*}
&\{a'_1 = a_1, a'_{12} = a_{12}, a'_{123} = -a_{123}, a'_2 = a_2, a'_3 = -a_3\}\ ,\\
&\{a'_1 = -a_1, a'_{12} = -a_{12}, a'_{123} = a_{123}, a'_2 = -a_2, a'_3 = a_3\}\ ,\\
&\{a'_1 = -a_1, a'_{12} = a_{12}, a'_{123} = -a_{123}, a'_2 = -a_2, a'_3 = -a_3\}\ ,
\end{align*}
respectively. It is clear that each of these solutions has at least one negative variable (given that all $a_\alpha$ are positive).
The 5th solution, however, is more involved, and it is not immediately clear if any of the individual variables is out of the interval $(0,1]$
for all choices of $a_{1},a_{2},a_{12},a_{3},a_{123}$. Yet, we show that the product of two of the resulting variables
is always negative:\\
{\tt > eval(a1'$\cdot$a2', Sols[5]);}\\
\mbox{\ \ \ \ \ \ \ \ \ \ \ \ \ \ \ \ \ \ \ \ \ \ \ \ \ \ \ \ \ \ \ \ \ \ \ \ \ \ \ \ \ \ }{\tt -a123a3}\\
Namely the product of the values assigned to the two variables, $a'_1a'_2$, equals $-a_{123}a_3$\footnote{The two components of this solution that establish the result above
 can be found in the appendix in full details.}. This
establishes the claim that only the first solution is legal.
This means that, within our domain, for given $a_{1},a_{2},a_{12},a_{3},a_{123}$,
there is only solution for the set of equations, and the solution is these $a_{1},a_{2},a_{12},a_{3},a_{123}$ 
themselves. So,
for $0\leq q_1, q_2, q_{12}, q_3, q_{123}<\infty$, the mapping 
$(q_1, q_2, q_{12}, q_3, q_{123}) \mapsto (s_1, s_2, s_{12}, s_3, s_{123})$ is indeed one to one, 
as desired. \hfill$\Box$

\subsection{The Resolved Tree: Explicit Solutions}
The result in the previous subsection is {\em existential}. It says that for any $(s_1, s_2, s_{12}, s_3, s_{123})$ in the relevant range there is a unique $(q_1, q_2, q_{12}, q_3, q_{123})$, but provides no {\em explicit} method to produce the $q_{\alpha}$. Here, we 
employ the Hadamard transform to derive quadratic invariants that enable such construction. The drawback of this approach is
that the number of possible solutions is greater than $1$, and can be up to $4$. We know that only one will be in the
relevant range, but a-priori do not know which one it is.

The Hadamard transformation ${\bf s} = {\bf s(q)} = H^{-1}\exp(H{\bf q})$ is one-to-one, so it can also be expressed as a mapping 
from {\bf s} to {\bf q}. This allows us to express each $q_{\alpha}$ in terms of the 7 values 
$(s_1, s_2, s_{12}, s_3, s_{13}, s_{23}, s_{123})$ ($s_{\emptyset}$ equals 1 minus the sum of these 7 $s_{\alpha}$). For the tree $T$ of Fig.~\ref{treeEx1}, two splits are not realized by 
any edge of $T$ and this corresponds to the identities $q_{13}=q_{23}=0$. This gives rise to two invariants involving the 7 $s_{\alpha}$. After some manipulation and simplification, we derive the following two invariants \cite{CHHP}:
\begin{align*}
0&=1-2\,{\it s_{1}}-2\,{\it s_{2}}-2\,{\it s_{3}}-2\,{\it s_{123}}\\
&- \left( -1+2\,{
\it s_{1}}+2\,{\it s_{2}}+2\,{\it s_{13}}+2\,{\it s_{23}} \right)  \left( -1+2\,{
\it s_{3}}+2\,{\it s_{13}}+2\,{\it s_{23}}+2\,{\it s_{123}} \right)
\end{align*}
\begin{align*}
0&=\left( -1+2\,{\it s_{1}}+2\,{\it s_{12}}+2\,{\it s_{13}}+2\,{\it s_{123}}
 \right)  \left( -1+2\,{\it s_{2}}+2\,{\it s_{12}}+2\,{\it s_{3}}+2\,{\it s_{13}}
 \right)\\
 & - \left( -1+2\,{\it s_{2}}+2\,{\it s_{12}}+2\,{\it s_{23}}+2\,{\it 
s_{123}} \right)  \left( -1+2\,{\it s_{1}}+2\,{\it s_{12}}+2\,{\it s_{3}}+2\,{\it 
s_{23}} \right)
\end{align*}
From the first equation, we get a quadratic equation for $x=s_{13}+s_{23}$,
\begin{align*}
0&=1-2\,{\it s_{1}}-2\,{\it s_{2}}-2\,{\it s_{3}}-2\,{\it s_{123}}\\
&- \left( -1+2\,{
\it s_{1}}+2\,{\it s_{2}}+2\,x \right)  \left( -1+2\,{
\it s_{3}}+2\,{\it s_{123}}+2\,{\it x} \right)
\end{align*}
Once we obtain a solution for $x=s_{13}+s_{23}$, the second invariant can be rewritten as 
a quadratic equation in $y=s_{13}-s_{23}$, where all coefficients are known
\begin{align*}
0=&\left( -1+2\,{  s_{1}}+2\,{  s_{12}}+2\,{  s_{123}+x+y}
 \right)\\
 &\ \ \cdot  \left( -1+2\,{  s_{2}}+2\,{  s_{12}}+2\,{  s_{3}}+x +y
 \right)\\
 &-\left( -1+2\,{  s_{2}}+2\,{  s_{12}}+2\,{  s_{123}}+{  x -y} \right)\\
&\ \ \cdot  \left( -1+2\,{  s_{1}}+2\,{  s_{12}}+2\,{  s_{3}}+{  x-y} \right)\ .
\end{align*}
In summary, we have up to two solutions for $x=s_{13}+s_{23}$, and for each of them up to 2 solutions for $y=s_{13}-s_{23}$.
Once $s_{13},s_{23}$ are determined, we find $s_{\emptyset}$ using the invariant
$$s_{\emptyset}+s_1+s_2+s_{12}+s_3+s_{13}+s_{23}+s_{123}=1\ .$$

\subsection{The Star Tree on Four Leaves}
The star tree is a special case of the resolved tree on four nodes, where the internal edge $q_{12}$ is of length 0.
However, the uniqueness result we have for the resolved tree does {\em not} directly imply a similar one for the star, as
we obtained one fewer observed pattern probability. We want to show that in this case, the mapping $(q_1, q_2, q_3, q_{123}) \mapsto (s_1, s_2, s_3, s_{123})$ is one to one. While this does not directly follow from the previous result, a similar algebraic approach
does work here as well.
\begin{align*}
s_1&= a_{123}a_{3} - a_{1}a_{2} + a_{2}a_{3} + a_{123}a_{2} - a_{1}a_{3} - a_{1}a_{123} - a_{1}a_{123}a_{2}a_{3} + 1\\
s_2&= a_{123}a_{3} - a_{1}a_{2} - a_{2}a_{3} - a_{123}a_{2} + a_{1}a_{3} + a_{1}a_{123} - a_{1}a_{123}a_{2}a_{3} + 1\\
s_3&= -a_{123}a_{3} + a_{1}a_{2} - a_{2}a_{3} + a_{123}a_{2} - a_{1}a_{3} + a_{1}a_{123} - a_{1}a_{123}a_{2}a_{3} + 1\\
s_{123}&= -a_{123}a_{3} + a_{1}a_{2} + a_{2}a_{3} - a_{123}a_{2} + a_{1}a_{3} - a_{1}a_{123} - a_{1}a_{123}a_{2}a_{3} + 1
\end{align*}

Again, we preprocess the system by computing a few linear combinations of the $s_\alpha$.
\begin{align*}
&(s_{1}+s_{2})/2=-a_1a_{123}a_2a_3 - a_1a_2 + a_{123}a_3 + 1\\
&(s_{1}+s_{3})/2=
-a_1a_{123}a_2a_3 - a_1a_3 + a_{123}a_2 + 1\\
&(s_{1}+s_{2}+s_{3}+s_{123})/4=-a_1a_{123}a_2a_3 + 1
\end{align*}
The resulting system of four polynomial equations (of degree 4 though, as before) is now
\begin{align*}
\mbox{\tt (eq1)\ \ \ } & a_{123}a_{3} - a_{1}a_{2} + a_{2}a_{3} + a_{123}a_{2} - a_{1}a_{3} - a_{1}a_{123} - a_{1}a_{123}a_{2}a_{3} \\
&= a'_{123}a'_{3} - a'_{1}a'_{2} + a'_{2}a'_{3} + a'_{123}a'_{2} - a'_{1}a'_{3} - a'_{1}a'_{123} - a'_{1}a'_{123}a'_{2}a'_{3} \\
\mbox{\tt (eq2)\ \ \ }  &-a_1a_{123}a_2a_3 - a_1a_2 + a_{123}a_3=-a'_1a'_{123}a'_2a'_3 - a'_1a'_2 + a'_{123}a'_3 \\
\mbox{\tt (eq3)\ \ \ }  &-a_1a_{123}a_2a_3 - a_1a_3 + a_{123}a_2=-a'_1a'_{123}a'_2a'_3 - a'_1a'_3 + a'_{123}a'_2 \\
\mbox{\tt (eq4)\ \ \ }  &-a_1a_{123}a_2a_3=-a'_1a'_{123}a'_2a'_3 \\
\end{align*}
\noindent
{\tt > with(SolveTools):}\\
{\tt > Sols := PolynomialSystem(\{eq{1}, eq{2}, eq{3}, eq4\},}\\ 
\mbox{\ \ \ \ \ \ \ \ \ \ \ \ \ \ \ \ \ \ \ }{\tt \{a1',a2',a3',a123'\}, domain = real):}\\
\noindent
Maple (version 17) produced 12 sets of solutions  to this system. 
Again, the first solution,
$$\{a'_1 = a_1,  a'_{123} = a_{123}, a'_2 = a_2, a'_3 = a_3\}\ ,$$ is exactly what we are looking for. For the other
sets of solutions, each either has a variable assigned to a negative value, {\em e.g.}\\
{\tt > eval(a3', Sols[2]);}\\
\mbox{\ \ \ \ \ \ \ \ \ \ \ \ \ \ \ \ \ \ \ \ \ \ \ \ \ \ \ \ \ \ \ \ \ \ \ \ \ \ \ \ \ \ }{\tt -a3}\\
or it is a root of a quadratic equation, {\em e.g.} \\
{\tt > eval(a3', Sols[6]);}\\
\mbox{\ \ \ \ \ \ \ \ \ \ \ \ \ \ \ \ \ \ \ \ \ \ \ \ \ \ \ \ \ \ \ \ \ \ \ \ \ \ \ \ \ \ }{\tt RootOf(a1*\_Z$^\wedge$2+a123*a2*a3)}\ ,\\
namely the solution of $a_1z^2+a_{123}a_2a_3=0$, which is not a real number. Therefore, all other solutions are illegal (not in the (0,1] interval).
 \hfill$\Box$

\subsection{The Star Tree: Explicit Solutions}
In case of the star tree on $m=4$ leaves,  only the four edges with pendant leaves, $q_1, q_2, q_3, q_{123}$, 
could have non-zero lengths. The splits corresponding to internal edges are not present and  in this tree and this leads to the identities: 
$q_1, q_2, q_3=0$. Expressing these edges in terms of {\bf s}, slightly manipulating and simplifying the expressions, we get the following three quadratic  invariants:
\begin{align*}
(1)\   &\ (-1+2 s_1+2 s_2+2  s_3+2  s_{123})\\
&=(-1+2 s_1+2  s_{12}+2   s_{13}+2   s_{123})\cdot (-1+2  s_2+2   s_{12}+2  s_3+2   s_{13})\ .
\end{align*}
\begin{align*}
(2) \ &\ (-1+2 s_1+2 s_2+2  s_3+2  s_{123})\\
&=(-1+2 s_2+2 s_{12}+2 s_{23}+2 s_{123}) \cdot (-1+2 s_1+2 s_{12}+2 s_3+2 s_{23})\ .
\end{align*}
\begin{align*}
(3) \ &\ (-1+2 s_1+2 s_2+2  s_3+2  s_{123})\\
&=(-1+2 s_1+2 s_2+2 s_{13}+2 s_{23}) \cdot (-1+2 s_3+2 s_{13}+2 s_{23}+2 s_{123})\ .
\end{align*}

\noindent
Let $x=s_{12}+s_{13}$, $y=s_{12}+s_{23}$, and $z=s_{13}+s_{23}$. Note that
$x+y+z=2s_{12}+2s_{13}+2s_{23}$. The three equations  (1)--(3) become three quadratic equations,
each in one of these variables, which can therefore be solved separately.   Recall that $s_1, s_2, s_3, s_{123}$ are known.
\begin{align*}
(1)\   &\ (-1+2 s_1+2 s_2+2  s_3+2  s_{123})\\
&=(-1+2 s_1+2   s_{123}+2x)\cdot (-1+2  s_2+2  s_3+2x)\ .
\end{align*}
\begin{align*}
(2) \ &\ (-1+2 s_1+2 s_2+2  s_3+2  s_{123})\\
&=(-1+2 s_2+2 s_{123}+2y) \cdot (-1+2 s_1+2 s_3+2y)\ .
\end{align*}
\begin{align*}
(3) \ &\ (-1+2 s_1+2 s_2+2  s_3+2  s_{123})\\
&=(-1+2 s_1+2 s_2+z) \cdot (-1+2 s_3+2 s_{123}+z)\ .
\end{align*}
\subsection{The star tree and the Inverse Function Theorem}
Let $T_m$ denote the star tree on $m$ leaves.  If $p_i$ denotes the probability of a state change ($0$ to $1$ or visa versa) on the edge incident with leaf $i$, $i=1,2,3,4$, then the probability $s_i$ of the obtaining the tree split pattern that separates leaf $i$ from the other leaves (i.e. the probability that leaf $i$ has a different state to all the other leaves) is given, in terms of $p_1,p_2,p_3,p_4$ ($0 \leq p_i < 1/2$) by:
\begin{equation}
\label{sieq}
s_i = p_i \prod_{j \neq i}(1-p_j) + (1-p_i)\prod_{j \neq i}p_j.
\end{equation}
This formula is easily verified by observing that there are two way to obtain the tree split that separates leaf $i$ from the other leaves -- either there is a change on the edge incident with leaf $i$, and no changes
on any of the other edges (the first term), or there is no change on the leaf incident with leaf $i$ but there are changes on all the other edges (the second term).

In the special case with $m=4$, the  4-by-4  Jacobian  matrix $${\bf J} = \left[ \frac{\partial s_i}{\partial p_j}\right]_{1\leq i,j\leq 4}\ .$$
Suppose ${\bf J}$ is invertible (the determinant is non zero) in a certain domain, and the underlying functions
($s_1,s_2,s_3,s_4$ in our case) are differentiable. Then by the multivariate inverse function theorem, these functions are locally invertible. 

We used Maple to compute the Jacobian and its determinant. It turns out that the determinant can be expressed 
as a product of three simple multivariate polynomials:
\begin{align*}
\mbox{det}({\bf J})=&\left( 2\,{\it p_1}\,{\it p_2}-{\it p_1}+1-{
\it p_2}-{\it p_3}+2\,{\it p_3}\,{\it p_4}-{\it p_4} \right)\\
\cdot&\left( 2\,{\it p_1}\,{\it p_3}-{\it p_1}+1-{\it p_2}-{\it p_3}+2\,{\it p_2}
\,{\it p_4}-{\it p_4} \right)\\
\cdot &\left( 2\,{
\it p_1}\,{\it p_4}-{\it p_1}+1-{\it p_2}+2\,{\it p_2}\,{\it p_3}-{\it p_3}-{
\it p_4} \right) 
\end{align*}
Each factor has a similar structure, and can be represented as follows:
\begin{align*}
&\left( 2{\it p_1}{\it p_2}-{\it p_1}+1-{\it p_2}-{\it p_3}+2{\it p_3}{\it p_4}-{\it p_4} \right)\\
=&\frac{1}{2}(2p_1-1)\cdot(2p_2-1)+\frac{1}{2}(2p_3-1)\cdot(2p_4-1)\ .
\end{align*}
It is clear that if all $p_i$ are in the range $0 \leq p_i < 1/2$, then each factor is strictly positive, so the determinant
is positive, as desired.  So by the multivariate inverse function theorem, under these 
conditions, the functions are locally invertible for any neighborhood in the domain
$0 \leq p_1,p_2,p_3,p_4 < 1/2$.  By contrast,  the invertibility results in section 
4.4 {\em are} global. 

\section{Concluding Comments}

We have made the first steps towards  settling the question  of whether the probability distribution of tree splits suffices to determine the branch lengths in the 2-state symmetric model.  We have shown that this holds in two cases:
\begin{itemize}
\item  for any tree if the branch lengths are sufficiently short, and 
\item the branch lengths are strictly positive and the tree has at most four leaves.  
\end{itemize} 
We conjecture that invertibility holds for any phylogenetic tree (with any number of leaves) and across the space of all non-negative branch lengths, however a proof will require
a different or modified approach. The algebraic approach may be extended to slightly larger 
numbers of taxa, but the underlying complexity of solving such systems of polynomial 
equations is quite prohibitive as the number of taxa grows. For the approach employing the inverse function theorem to be applicable, it will be required to show global, rather than 
just local, invertibility.

Note that a related but simpler question  asks if any  $k$ linear combinations of the site pattern probabilities identifies the branch lengths.  In this case, the answer is `yes'  for any number of leaves (and any tree) and this can be seen by combining three observations that apply for a wide range of substitution models (not just on two states):  (i) the expected probability that
any pair of $i,j$ differ in state is a sum of certain pattern probabilities, (ii)  this expected probability can be transformed to give the sum of the branch lengths between $i$ and $j$ in the tree \cite{fels}, and (iii) $k$ such (carefully selected)  pairs of path distances suffice to recover the length of all the $k$ edges \cite{dre}.

Notice also that if instead of branch lengths we parameterize by the vector of probabilities of state changes on edges, ${\bf p} = [p_e]_{e \in E(T)}$ then for invertibility to hold in 
$[0, \eta']^{E(T)}$ it is necessary that $\eta' < \frac{1}{2}$. To see this, consider any phylogenetic tree for which one edge $e$ has $p_e=x \in (0,\frac{1}{2})$, while for all other edges $f \neq e$ we have
$p_f = \frac{1}{2}$.  Then for every edge of the  tree,  the tree split that corresponds to that edge has probability $2^{1-m}$ where $m$ is the number of leaves of the tree. Since this distribution is invariant to variation in $x$ we see that invertibility fails in this case. Thus the invertibility conjecture is more straightforward to state when phrased within the context of branch lengths than probabilities of state changes on edges. 

\section{Acknowledgments}
This work was initiated while B.C. was on a sabbatical visit at the University of Canterbury, New Zealand.
M.S. thanks the Allan Wilson Centre for Molecular Ecology and Evolution for funding support. 
The work of B.C. was not supported by the ISF (Israeli Science Foundation) or any other funding agency.

\section*{Appendix}
\subsection*{Details of the 5th Solution for the Resolved Tree:}$ $\\
We obtained
\begin{gather*}
a'_1={\tt RootOf} \left(  \left(  a_1 a_2- a_{12} a_{123}
 a_2- a_{12} a_2 a_3+ a_{123} a_3 \right)\_Z^2
+ a_1 a_{123} a_2 a_3\right. \\
\left. - a_1 
a_{12} a_{123}^2 a_3- a_1 a_{12} a_3^2
 a_{123}+ a_{123}^2 a_3^2 \right) \ , \\
a'_2=   -a_{123} a_3\, {\tt RootOf}^{-1}\left(  \left( - a_{12}
 a_{123} a_2- a_{12} a_2 a_3+ a_1 a_2
+ a_{123} a_3 \right)  \_Z^{2}\right. \\
\left. - a_1 a_{12}
 a_{123}^2 a_3- a_1 a_{12} a_{123} a_3^2
+ a_1 a_{123} a_2 a_3+ a_{123}^2 a_3^
2 \right)
\end{gather*}
where the root in both expressions correspond to the same value.
Multiplying the two solutions, we conclude that $a'_1a'_2=-a_{123}a_3$.

\subsection*{An alternative approach}
We present here an alternative strategy for deriving the quadratic invariants in the case of the resolved tree on
4 leaves. We use the five tree split probabilities to determine the probabilities of the remaining three patterns 
since, once we have all eight pattern probabilities,  the branch lengths can be easily recovered. Such a strategy looks plausible since we have three equations involving the 8 split probabilities, the linear invariant $\sum_\alpha s_\alpha=1$ together with two quadratic invariants of the form
$p_1({\bf s})=0$ and $p_2({\bf s})=0$, which we describe shortly.

We can proceed as follows.   Firstly, to find $s_\emptyset$  let $E$ be the event that leaf $x_1$ is in the same state as leaf 
$x_2$, and let $F$ be the event that leaf $x_3$ is in the same state as leaf $x_4$.
Under the symmetric 2-state model (or indeed any group-based model on any number of states \cite{sem}) $E$ and $F$ are 
independent, and so $$p_1({\bf s}) = \PP(E \& F) -\PP(E)\PP(F)=0.$$
Now, notice that $E$ and $F$ each occur for just four patterns, and these patterns are either tree splits or the pattern 
$\alpha = \emptyset$ (where all leaves are in the same state).  Similarly $E \&F$ occurs for just two patterns, namely 
the tree split corresponding to the central edge, and $\alpha = \emptyset$.  Thus the equation $p_1({\bf s}) = 0$ provides 
a quadratic equation for $s_\emptyset$ whose other terms involve just tree split probabilities.  Thus $s_\emptyset$ can then 
be expressed in terms of (just) the tree split probabilities [Minor technical point to check:  a quadratic has two solutions - 
check only one of them is viable for $0\leq p_e< 1/2$ -- e.g. $s_\emptyset$ must be greater than any tree split probability].

Specifically, we have: $$\PP(E\&F) = s_\emptyset + s_{12},$$
$$\PP(E) = s_{\emptyset} + s_{12} + s_{123} + s_3 \mbox { and }  \PP(F) = s_\emptyset+s_{12} + s_{1} + s_{2}.$$

Thus if we set $a=s_{123}+s_3$, $b= s_{1}+s_{2}$ (both of which are `known' since they are sums of probabilities of trees splits) then we obtain the quadratic equation for $\theta: = s_\emptyset + s_{12}$.

$$\theta = (\theta+a)(\theta+b).$$
Thus, $$s_\emptyset = -s_{12} + \frac{1}{2}\left (1-a-b \pm \sqrt{(1-a-b)^2 - 4ab}\right ).$$

Once we have $s_\emptyset$ we can calculate the probability of the sum $s'$ of the two non-tree splits by the linear invariant $\sum_\alpha s_\alpha=1$ to obtain:

$$s' = 1-S-s_\emptyset,$$
where $S = s_1+s_2+s_3+s_{123}+s_{12}$ is the sum of the five tree splits. 

Let $x=s_{13}$.  Then $s_{23} = s'-x$, and provided we can calculate $x$ (using the tree splits, and the new values we have derived from them, namely $s_\emptyset$ and $s'$) then we are done. 

To calculate $x$ we consider a second quadratic
invariant 
\begin{equation} 
\label{rreq}
p_2({\bf s})=(1-2p_{13})(1-2p_{24}) - (1-2p_{14})(1-2p_{23})=0,
\end{equation}

where $p_{ij}$ is the sum of the pattern probabilities for which leaf $x_i$ is in a different state to leaf $x_j$.

The justification of Eqn.~\ref{rreq} follows immediately from observing that $(1-2p_{ij}) = \exp(-q_{ij})$ where $q_{ij}$ is the sum of the (additive) edge lengths on the path connecting the leaves $i$ and $j$ in $T$.

Notice that $x$ contributes to each term in Eqn. (\ref{rreq})  in the first product for $p_2({\bf s})$ via $s_{23} = s'-x$ and it appears directly in both terms of the second product in $p_2({\bf s})$ and so
$p_2({\bf s})=0$ gives us a quadratic equation for $x$ in terms of quantities that are either known (the tree split probabilities) or have been determined ($s_\alpha$ and $s'$). 

Specifically, we have: 

$$p_{13}= s_1+s_3 + s_{12}+s_{23}; \mbox{ } p_{24} = s_2 + s_{12} + s_{23}+ s_{123};$$
$$p_{14}= s_{12}+ s_{13}+ s_{123}+ s_1; \mbox{ } p_{23} = s_2 + s_3 +s_{12}+ s_{13}.$$

So if we write Eqn. (\ref{rreq}) as:
$$p_{13}+p_{24}-p_{14}-p_{23} = 2(p_{13}p_{24} - p_{14}p_{23})$$
and since the left-hand side is $2(s_{23} - s_{13})$ we obtain:
$$s' - 2x = p_{13}p_{24} - p_{14}p_{23}.$$
where each term on the  right-hand side of this equation also involves $x$ linearly (and so we obtain a quadratic equation for $x$).

\end{document}